\documentclass[10pt,conference]{IEEEtran}
\pdfoutput=1

\usepackage[latin9]{inputenc}
\usepackage[T1]{fontenc}

\usepackage{amsthm}
\usepackage{amsmath}
\usepackage{amssymb}
\usepackage{commath}              
\usepackage{graphicx}
\usepackage{color}
\usepackage{url}
\usepackage{algorithmic}
\usepackage{algorithm}
\usepackage{authblk}

\makeatletter                                

\numberwithin{equation}{section}
\numberwithin{figure}{section}

\theoremstyle{plain}
\newtheorem{thm}{\protect\theoremname}[section]
\theoremstyle{definition}
\newtheorem{defn}[thm]{\protect\definitionname}
\theoremstyle{remark}

\theoremstyle{plain}
\newtheorem{lem}[thm]{\protect\lemmaname}
\theoremstyle{remark}

\theoremstyle{plain}
\newtheorem{corollary}[thm]{\protect\corollaryname}
\theoremstyle{plain}

\theoremstyle{remark}

\providecommand{\claimname}{Claim}
\providecommand{\definitionname}{Definition}
\providecommand{\lemmaname}{Lemma}
\providecommand{\remarkname}{Remark}
\providecommand{\theoremname}{Theorem}
\providecommand{\corollaryname}{Corollary}
\providecommand{\propositionname}{Proposition}
\providecommand{\examplename}{Example}

\newcommand{\R}{\mathbb{R}}
\newcommand{\Z}{{\mathbb{Z}}}
\newcommand{\E}{\mathbb{E}}

\DeclareMathOperator*{\argmin}{argmin}
\DeclareMathOperator*{\argmax}{argmax}
\DeclareMathOperator*{\tr}{tr}
\DeclareMathOperator*{\Cov}{Cov}
\newcommand{\N}{{\mathcal{N}}}

\newcommand{\MSE}{\mathrm{MSE}}
\newcommand{\kk}{{\mathbf{k}}}
\newcommand{\OO}{{\mathcal{O}}}

\newcommand{\LN}{{\mathcal{L}}_N}

\newcommand{\A}{{\mathcal{A}}}

\let\P\PSymbol
\def\P{\mathbb{P}}

\newcommand{\error}{Z}

\newcounter{nienumi}
\def\niitem{\par\smallskip \textit{\refstepcounter{nienumi}{\arabic{nienumi}})\hspace{.5em}}}
\newenvironment{nienumerate}
{\parindent0pt\par}
{\setcounter{nienumi}{0}\par\smallskip}

\begin{document} 

\title{Estimation in the group action channel}

\author[1,2]{Emmanuel Abbe}
\author[1]{Jo\~ao M. Pereira}
\author[1,3]{Amit Singer}

\thanks{
	EA was partly supported by the NSF CAREER Award CCF--1552131, ARO grant W911NF--16--1--0051 and NSF Center for the Science of Information CCF--0939370.  JP and AS were partially supported by Award Number R01GM090200 from the NIGMS, the Simons Foundation Investigator Award and Simons Collaborations on Algorithms and Geometry, the Moore Foundation Data-Driven Discovery Investigator Award, and AFOSR FA9550-17-1-0291.}
 
 \affil[1]{The Program in Applied and Computational Mathematics,
    Princeton University, Princeton, NJ, USA}
\affil[2]{Electrical Engineering Department, Princeton University, Princeton, NJ, USA}
\affil[3]{Department of Mathematics, Princeton University, Princeton, NJ, USA\vspace{-.2cm}
}

\maketitle

\begin{abstract}
We analyze the problem of estimating a signal from multiple measurements on a \emph{group action channel} that linearly transforms a signal by a random group action followed by a fixed projection and additive Gaussian noise. This channel is motivated by applications such as multi-reference alignment and cryo-electron microscopy. We focus on the large noise regime prevalent in these applications. We give a lower bound on the mean square error (MSE) of any asymptotically unbiased estimator of the signal's orbit in terms of the signal's moment tensors, which implies that the MSE is bounded away from 0 when $N/\sigma^{2d}$ is bounded from above, where $N$ is the number of observations, $\sigma$ is the noise standard deviation, and $d$ is the so-called \emph{moment order cutoff}. In contrast, the maximum likelihood estimator is shown to be consistent if $N /\sigma^{2d}$ diverges.
\end{abstract}

\begin{IEEEkeywords}
Multi-reference alignment, cryo-EM, Chapman-Robbins bound.
\end{IEEEkeywords}


\section{Introduction}
In this paper, we consider the problem of estimating $x\in \R^L$ with $N$ measurements from the \emph{group action channel}, defined as
\begin{equation}\label{eq:mra}
Y_j  = P G_j x + \sigma \error_j\in \R^K,\quad j= 1,\dots,N,
\end{equation}
where the $\error_j$ are i.i.d and drawn from $\N(0,I_K)$, i.e. $\error_j\in \R^K$ and its entries are i.i.d standard Gaussian variables; $P\in \R^{K\times L}$ is a projection matrix which is known; $G_j$ are i.i.d. matrices drawn from a distribution $\theta$ on a compact subgroup $\Theta$ of $O(L)$, i.e. the space of orthogonal matrices in $\R^L$. The distribution $\theta$ is not known, however the main goal is to estimate the signal $x\in \R^L$.

The goal of this paper is to understand the sample complexity of \eqref{eq:mra}, i.e. the relation between the number of measurements and the noise standard deviation such that an estimator $\hat X$, of $x$, converges in probability to the true value with $N$ diverging, up to a group action. Allowing for a group action is intrinsic to the problem: if we apply an element $g$ of $\Theta$ to $x$, and its inverse $g^{-1}$ to the right of $\theta$, we will produce exactly the same samples, thus there is no estimator $\hat X$ that is able to distinguish the observations that originate from $x$ and the ones from $gx$.

The model \eqref{eq:mra} is a generalization of multi-reference alignment (MRA), which arises in a variety of engineering and scientific applications, among them structural biology~\cite{park2011stochastic,park2014assembly,scheres2005maximum}, radar~\cite{zwart2003fast,gil2005using}, robotics~\cite{rosen2016certifiably} and image processing~\cite{dryden1998statistical,foroosh2002extension,robinson2009optimal}.
The one-dimensional MRA problem, where $\Theta$ is the group generated by the matrix $R$ that cyclically shifts the elements of the signal, i.e. it maps $(x_1,\dots,x_L)\mapsto (x_L,x_1\dots,x_{L-1})$, has been recently a topic of active research. In \cite{bandeira2017optimal}, it was shown that the sample complexity is $\omega(\sigma^{6})$ when $\theta$ is the uniform distribution and the projection matrix is the identity. In \cite{perry2017sample} it was presented a provable polynomial time estimator that achieves the sample complexity, while in  \cite{bendory2017bispectrum} it was presented a non-convex optimization framework that is more efficient in practice. Note that, when the projection matrix is the identity, we can always enforce a uniform distribution on $\Theta$ by applying a random group action, i.i.d. and drawn from the uniform distribution, to the observations. In \cite{abbe2017multireference}, it was shown that $\omega(\sigma^{6})$ is also the sample complexity if $\theta$ is unknown beforehand but is uniform or periodic, this is, $\theta=R^{\ell}\theta$ for some $1\le \ell\le L-1$. However, if $\theta$ is aperiodic, the sample complexity is $\omega(\sigma^{4})$. It is also presented an efficient estimator that uses the first and second moments of the signal over the group, which can be estimated with order of $\sigma^{2}$ and $\sigma^{4}$ observations, respectively, thus achieving the sample complexity. The main result in this paper is a generalization of the information lower bound presented in \cite{abbe2017multireference}, however the proof techniques remain the same. 

We can also use \eqref{eq:mra} to model the problem of single particle reconstruction in cryo-electron microscopy (cryo-EM), in which a
three-dimensional volume is recovered from two-dimensional
noisy projections taken at unknown viewing directions \cite{bartesaghi20152,sirohi20163}. Here $x$ is a linear combination of products of spherical harmonics and radial basis functions, $\Theta\equiv SO(3)$, and its elements act on $x$ by rotating the basis functions. Finally, $P$ is a tomographic projection onto the $xy$ plane. The paper \cite{bandeira2017estimation} considers the problem \eqref{eq:mra} with $\theta$ being known and uniform. It obtains the same result for the sample complexity as this paper, and together with results from computational algebra and invariant theory verifies that in many cases the sample complexity for the considered cryo-EM model is $\omega(\sigma^6)$, and at least $\omega(\sigma^6)$ more generally. They also consider the problem of heterogeneity in cryo-EM.

\section{The Main Result}

Since we can only determine $x$ up to a group action, we define  the best alignment of $\widehat X$ with $x$ by 
\begin{equation}\label{eq:Rxdef}
\phi_{x}(\widehat{X}) = \argmin_{z\in \{ g\widehat{X}\}_{g\in\Theta}} \|z-x\| .
\end{equation}
and the mean square error (MSE) as
\begin{align}
\nonumber \MSE&:=\E\left[\min_{g\in\Theta} \|g \widehat X-x\|^2\right],\\
\label{eq:MSE+Rx} &=\E\left[\|\phi_{x}(\widehat X)-x\|^2\right].
\end{align}
The expectation is taken over $\widehat X$, which is a function of the observations with distribution determined by \eqref{eq:mra}. Since we are interested in estimators that converge to an orbit of $x$ in probability as $N$ diverges, we only consider estimators which are asymptotically unbiased, i.e., $\E[\phi_{x}(\widehat X)]\rightarrow x$ as $N \to \infty$. However the results presented in this paper can be adapted to biased estimators (see Theorem~\ref{thm:ChapmanRobbins}).

Let us introduce some notation regarding tensors.  For a vector $x\in \R^L$, we denote by $x^{\otimes n}$ the $L^{\otimes n}$ dimensional tensor where the entry indexed by $\kk=(k_1,\dots,k_n)\in \Z^n_L$ is given by $\prod_{j=1}^n x[k_j]$. The space of $n$-dimensional tensors forms a vector space, with sum and multiplication defined entry-wise. This vector-space has inner product and norm defined by $\left<A,B\right>=\sum_{\kk\in\Z_L^n}A[\kk]B[\kk]$ and $\|A\|^2=\left<A,A\right>$, respectively.

\begin{defn}\label{defn:autocorr}
	The $n$-th order moment of $x$ over $\theta$, is the tensor of order $n$ and dimension $K^{n}$, defined by
	\begin{equation}\label{eq:autocorrelation}
	M^n_{x,\theta}:=\E\left[(PG x)^{\otimes n}\right],
	\end{equation}
	where $G\sim \theta$.	
	
\end{defn}

In this paper, we provide lower bounds for the MSE in terms of the noise standard deviation and the number of observations. We show that the MSE is bounded below by order of $N/\sigma^{2\bar d}$, where $\bar d$ is the \emph{moment order cutoff}, defined as the smallest such that the moment tensors up to order $\bar d$ define $x$ unequivocally.
We also show that if $N>>\sigma^{2\bar d}$, then the marginalized maximum likelihood estimator (MLE) converges in probability to the true signal (up to a group action). We now present the main result of the paper.

\begin{thm}\label{thm:corsigma}
Consider the estimation problem given by equation \eqref{eq:mra}. For any signal $\tilde x\in \R^L$ such $\phi_x(\tilde x)\neq x$ and for any group distribution $\tilde \theta$, let $K^n_{\tilde x,\tilde \theta}=\frac1{n!}\|M^n_{\tilde x,\tilde \theta}- M^n_{x,\theta}\|^2$, 
\linebreak $d_{\tilde x,\tilde \theta}=\inf\left\{n:K^n_{\tilde x,\tilde \theta}>0\right\}$ and
define the \textbf{moment order cutoff} as $\bar d=\max d_{\tilde x,\tilde \theta}$.
Finally let
$$\lambda^{m}_N=N/\sigma^{2m}.$$
We have
\begin{equation}\label{eq:momentssanov}
\MSE\ge \sup_{\tilde x,\tilde \theta: d_{\tilde x,\tilde \theta}=\bar d}\left\{\frac{\|\phi_{x}(\tilde x)-x\|^2}
{\exp\left(\lambda^{\bar d}_N K^{\bar d}_{\tilde x,\tilde \theta}\right)-1+\OO\left(\lambda^{\bar d}_N \sigma^{-1}\right)}\right\},
\end{equation}
thus the MSE is bounded away from zero if $\lambda^{\bar d}_N$ is bounded from above. Moreover, if $\lim_{N\rightarrow \infty}\lambda^{\bar d}_N=\infty$, then the MLE converges in probability to $gx$, for some element $g\in\Theta$.
\end{thm}

\subsection{Taking the limit $(\tilde x,\tilde \rho)\rightarrow (x,\rho)$}
Theorem \ref{thm:corsigma} is an application of a modified Chapman-Robbins bound, presented later in Theorem \ref{thm:ChapmanRobbins}. On the other hand the classical Cram\'er-Rao bound~\cite{CramerRao}, which gives a lower bound on the variance of an estimator $\hat S$ of a parameter $s\in \R$, can be obtained from the Chapman-Robbins bound by taking the limit $\tilde s\rightarrow s$. We present an analog version of Theorem \ref{thm:corsigma} obtained by taking a similar limit.

\begin{corollary}
	Under the conditions of Theorem \ref{thm:corsigma}, 
	let\linebreak $x_h:=(1-h)x+h\tilde x$, $\theta_h:=(1-h)\theta+ h\tilde \theta$,
	$$Q^n_{\tilde x,\tilde \theta}=\lim\limits_{h\rightarrow 0} \frac1{n!h^2}\|M^n_{x_h,\theta_h}- M^n_{x,\theta}\|^2,$$ 
	\smallskip $q_{\tilde x,\tilde \theta}=\inf\left\{n:Q^n_{\tilde x,\tilde \theta}>0\right\}$ and $\bar q=\max q_{\tilde x,\tilde \theta}$. 
	Then
	\begin{equation}\label{eq:momentscramerrao}
	\MSE\ge \sup_{\tilde x,\tilde \theta: q_{\tilde x,\tilde \theta}=\bar q}\left\{\frac{\|\phi_{x}(\tilde x)-x\|^2}{\lambda^{\bar q}_N Q^{\bar q}_{\tilde x,\tilde \theta}
		+\OO\left(\lambda^{\bar q}_N \sigma^{-1}\right)}\right\}.
	\end{equation}
\end{corollary}
We leave the proof of this corollary to \cite[Appendix C]{abbe2017multireference}. It is interesting to compare this bound with \eqref{eq:momentssanov} when $\lambda^{\bar d}_N$ or $\lambda^{\bar q}_N$ diverge. If $\bar q\ge \bar d$, then $\eqref{eq:momentscramerrao}$ will dominate $\eqref{eq:momentssanov}$, and the lower bound for the MSE will be inversely proportional to $\lambda^{\bar q}_N$, which is a behavior typical of estimation problems with continuous model paramaters. On the other hand, if $\bar d> \bar q$, then $\eqref{eq:momentssanov}$ dominates $\eqref{eq:momentscramerrao}$. the MSE will depend exponentially on $\lambda^{\bar d}_N$, which is a behaviour typical of discrete parameter estimation problems \cite{abbe2017sample}. One can show that $\bar d> \bar q$ only happens when the supremum in $\eqref{eq:momentssanov}$ is attained by some $x^*$ not in the orbit of $x$. The exponential decay in $\eqref{eq:momentssanov}$ is the same as the probability of error of the hypothesis testing which decides if the observations come from $x$ or $x^*$. 

We conjecture that the lower bounds presented in this paper can be achieved asymptotically by the MLE. In fact when the search space is discrete, the MLE achieves the least probability of error (assuming a uniform prior on the parameters), which behaves like \eqref{eq:momentssanov}. Also, when the search space is continuous, the MLE is asymptotically efficient, which means it achieves the Cramér-Rao lower bound. However this bound is obtained from the Chapman-Robbins lower bound (which we use in this paper) by taking a similar limit as in \eqref{eq:momentscramerrao}, and the bound also scales inversely proportional to the number of observations.

\subsection{Prior Knowledge}

The result presented can be adapted to improve the bound if we have prior knowledge about the signal and group distribution. If we know beforehand that $(x,\theta)\in\A$ (for instance, $x$ has a zero element or $\theta$ is the uniform distribution on $\Theta$), we can instead define 
$\bar d=\max_{(\tilde x,\tilde \theta)\in\A} d_{\tilde x,\tilde \theta}$ and restrict the supremum in \eqref{eq:momentssanov} to $(\tilde x,\tilde \theta)$ in $\A$.

\subsection{Examples}

\begin{nienumerate}
\niitem Let $x=(a,b,c)\in \R^3$; $\Theta$ be the group generated by the cyclic shift matrix $R$ that maps $(a,b,c)\mapsto (b,c,a)$; and $P$ projects $x$ into its first two elements, i.e $P(a,b,c)=(a,b)$. Furthermore, we know a-priori that one, and only one, of the elements of $x$ is $0$ (let's assume without loss of generality that $a=0$), the other two elements are distinct and $\theta$ is uniform, i.e. $\P(G=I)=\P(G=R)=\P(G=R^2)=\frac13$. We have
\begin{align*}
M^1_{x,\theta}&=\E\left[PG x\right],\\
&=\frac13(0,b)+\frac13(b,c)+\frac13(c,0),\\
&=\frac{b+c}3(1,1).
\end{align*}
and
\begin{align*}
M^2_{x,\theta}&=\E\left[(PG x)(PG x)^T\right],\\
&=\frac13\left[\begin{array}{cc}0&0\\0&b^2\end{array}\right]+
\frac13\left[\begin{array}{cc}b^2&bc\\bc&c^2\end{array}\right]+
\frac13\left[\begin{array}{cc}c^2&0\\0&0\end{array}\right],\\
&=\frac13\left[\begin{array}{cc}b^2+c^2&bc\\bc&b^2+c^2\end{array}\right].
\end{align*}

From these two moments, we can solve for $b$ and $c$, however all these equations are symmetric on $b$ and $c$, thus we can't identify which one of the values obtained is $b$ and which one is $c$. In other words, both candidate solutions are $x=(0,b,c)$ and $x^*=(0,c,b)$. However $M^3_{x,\theta}$ differs from $M^3_{x^*,\theta}$, if we look for the entry in $M^3_{x,\theta}$ indexed by $(1,1,2)$ we note that
\begin{align*}
M^2_{x,\theta}[1,1,2]&=\E\left[(PG x)^2_1(PG x)_2\right],\\
&=\frac13 0^2b+
\frac13b^2c+
\frac13c^20\\
&=\frac13 b^2c,
\end{align*}
and analogously $M^3_{x^*,\theta}[1,1,2]=c^2b$. From the 8 entries of $M^3_{x,\theta}$, $2$ are equal to $M^3_{x^*,\theta}$ and $6$ differ by $b^2c-c^2b$, in absolute value, so $\|M^3_{x^*,\theta}-M^3_{x,\theta}\|^2=6(b^2c-c^2b)^2$. This means $\bar d=3$, $\bar q=2$, thus if $\lambda_N^3$ the lower bound \eqref{eq:momentssanov} dominates \eqref{eq:momentscramerrao}, the supremum is attained at $x^*$ and
$$\MSE\ge \frac{\|\phi_{x}(x^*)-x\|^2}{\exp\left(\lambda_N^3(b^2c-c^2b)^2\right)-1+\OO\left(\lambda_N^3\sigma^{-1}\right)}.$$
Note that $\|\phi_{x}(x^*)-x\|^2=\min(b^2,c^2,2(b-c)^2)$.

\niitem Let $x=(a,b)\in \R^2$; $\Theta$ be the group generated by the cyclic shift matrix $R$ that maps $(a,b)\mapsto (b,a)$; and $P$ projects $x$ into its first element, i.e $P(a,b)=a$. Furthermore, we know a-priori that $\theta$ is uniform, i.e. $\P(G=I)=\P(G=R)=\frac12$. We have
\begin{align*}
M^1_{x,\theta}&=\E\left[PG x\right]=\frac12a+\frac12b=\frac{a+b}2
\end{align*}
and
\begin{align*}
M^2_{x,\theta}&=\E\left[(PG x)^2\right]=\frac12a^2+\frac12b^2=\frac{a^2+b^2}2.
\end{align*}
From these two moments we can determine $a$ and $b$ up to an action of the group. Now take $x_h=(a+h,b-h)$, so that $M^1_{x,\theta}=M^1_{x_h,\theta}$. We have
$$\lim_{h\rightarrow0}\frac1{h}(M^2_{x_h,\theta}-M^2_{x,\theta})=a-b.$$
Here $\bar q=\bar d=2$, thus if $\lambda_N^2$ diverges, \eqref{eq:momentscramerrao} dominates \eqref{eq:momentssanov}, and the lower bound is
$$\MSE\ge \frac{4}{\lambda_N^2(a-b)^2+
\OO\left(\lambda_N^2\sigma^{-1}\right)}.$$

\end{nienumerate}

\section{Proof Techniques} 

The outline of the proof is as follows. In Section~\ref{sec:ChapmanRobbins} we use an adaptation of the Chapman-Robbins lower bound~\cite{ChapRobb}, to derive a lower bound on the MSE in terms of the $\chi^2$ divergence, this is Theorem \ref{thm:ChapmanRobbins}. Then, in Section~\ref{sec:Fisherautocorr}, we express the $\chi^2$ divergence in terms of the Taylor expansion of the posterior probability density and the moment tensors, obtaining Lemma \ref{lem:falldominoes}. Finally in section \ref{sec:FinalDetails} we combine Theorem \ref{thm:ChapmanRobbins} and Lemma \ref{lem:falldominoes} to obtain \eqref{eq:momentssanov}, use Lemma \ref{lem:falldominoes} to obtain a similar Taylor expansion for the Kullback-Leibler (KL) divergence and use this to show that the MLE is consistent. 

Throughout the paper we denote the expectation by $\E$, use capital letter for random variables and lower case letter for instances of these random variables. Let $Y^{N}\in \R^{L\times N}$ be the collection of all measurements as columns in a matrix. Let us denote by $f^{N}_{x,\theta}$ the probability density of the posterior distribution of $Y^{N}$, 
\begin{equation}\label{eq:parindepence}
f^{N}_{x,\theta}(y^N)=\prod_{j=1}^N f_{x,\theta}(y_j),
\end{equation} 
and the expectation of a function $g$ of the measurements under the measure $f^{N}_{x,\theta}$ by
\begin{equation*}
\E_{x,\theta}\left[g\left(Y^{N}\right)\right]:=\int_{\R^{L\times N}} g\left(y^{N}\right) f^{N}_{x,\theta}\left(y^{N}\right)dy^{N}.
\end{equation*}
For ease of notation, we write $\E\left[g\left(Y^N\right)\right]$ when the signal and distribution are implicit. The bias-variance trade-off of the MSE is given by: \begin{equation}\label{eq:BiasVar}
\MSE=\tr(\Cov[\phi_x(\widehat X)])+
{\|\E[\phi_{x}(\widehat X)]-x\|^2},
\end{equation}
with
\begin{equation}\label{eq:Cov}
\Cov[\phi_x(\widehat X)]=\E\left[\phi_x(\widehat X)\phi_x(\widehat X)^T\right]-\E[\phi_{x}(\widehat X)]\E[\phi_{x}(\widehat X)]^T.
\end{equation}

Our last definition is of the $\chi^2$ divergence, which gives a measure of how "far" two probability distributions are.
\begin{defn}
	The $\chi^2$ divergence between two probability densities $f_A$ and $f_B$ is defined by
	\begin{equation*}
	\chi^2(f_A||f_B):=
	\E\left[\left(\frac{f_A(B)}{f_B(B)}-1\right)^2\right],
	\end{equation*}
	where $B\sim f_B$.
\end{defn}
Due to equation~\eqref{eq:parindepence}, the relation between the $\chi^2$ divergence for $N$ and one observations is given by
\begin{equation}\label{eq:nchi2}
 \chi^2(f^N_{\tilde x,\tilde \theta}||f^N_{x,\theta})=(1+\chi^2(f_{\tilde x,\tilde \theta}||f_{x,\theta}))^N-1.
\end{equation}


\subsection{Chapman-Robbins lower bound for an orbit}\label{sec:ChapmanRobbins}
The classical Chapman-Robbins gives a lower bound on an error metric of the form $\E[\|\widehat X -x\|^2]$, hence we modified it to accommodate to the group invariant metric defined in \eqref{eq:MSE+Rx}. We point out that $\Cov[\phi_x(\widehat X)]$ is related to the $\MSE$ by~\eqref{eq:BiasVar}.
\begin{thm}[Chapman-Robbins for orbits]\label{thm:ChapmanRobbins}
	For any $\tilde x\in\R^L$ and group distribution $\tilde \theta$ in $\Theta$, we have
	\begin{equation*}
	\Cov[\phi_x(\widehat X)]\succeq \frac{z z^T}
	{\chi^2(f^N_{\tilde x,\tilde \theta}||f^N_{x,\theta})},
	\end{equation*}
	where $z=\E_{\tilde x,\tilde \theta}[\phi_{x}(\widehat X)]-\E_{x,\theta}[\phi_{x}(\widehat X)]$. 
\end{thm}
\begin{proof}
The proof mimics the one of the classical Chapman and Robbins bound, and is also presented in \cite[Appendix A]{abbe2017multireference}. Define
$$V:=\frac{f^N_{\tilde x,\tilde \theta}(Y^N)}{f^N_{x,\theta}(Y^N)}.$$
and note that
\begin{itemize}
	\item $\E_{x,\theta}[g(Y^N)V]=\E_{\tilde x,\tilde \theta}[g(Y^N)],$ 
	\item $\E_{x,\theta}[V-1]=0,$
	\item $\E_{x,\theta}[(V-1)^2]=\chi^2(f^N_{\tilde x,\tilde \theta}||f^N_{x,\theta}).$
\end{itemize}
We have
\begin{align*}
w^T\left(\E_{\tilde x,\tilde \theta}[\phi_{x}(\widehat X)]-\E_{x,\theta}[\phi_{x}(\widehat X)]\right)\hspace{-98pt}\\
&=\E_{x,\theta}\left[w^T\left(\phi_x(\widehat X)-\E_{x,\theta}[\phi_{x}(\widehat X)]\right)(V-1)\right],
\end{align*}
and by Cauchy-Schwarz
\begin{multline*}
\left[w^T\left(\E_{\tilde x,\tilde \theta}[\phi_{x}(\widehat X)]-\E_{x,\theta}[\phi_{x}(\widehat X)]\right)\right]^2\\
\le \E_{x,\theta}[(w^T(\phi_x(\widehat X)-\E_{x,\theta}[\phi_{x}(\widehat X)]))^2]\chi^2(f^N_{\tilde x,\tilde \theta}||f^N_{x,\theta}).
\end{multline*}
\end{proof}

\subsection{$\chi^2$ divergence and moment tensors}
\label{sec:Fisherautocorr}
In this subsection we give a characterization of the $\chi^2$ divergence, which appears in the Chapman-Robbins bound, in terms of the moment tensors. 

Instead of considering the posterior probability density of $Y^N$, we will consider its normalized version $\widetilde Y^N=Y^N/\sigma$. We then have
\begin{equation}\label{eq:tildey}
\widetilde Y_j=\gamma P{G_j} x + \error_j,
\end{equation}
where $\gamma=1/\sigma$, $G_j\sim \theta$ and $\error_j\sim \N(0,I)$. While this change of variable does not change the $\chi^2$ divergence, we can now take the Taylor expansion of the probability density around $\gamma=0$, that is,
\begin{equation}\label{eq:probdenstaylor}
f_{x,\theta}(y;\gamma)=f_\error(y)\sum_{j=0}^\infty \alpha^j_{x,\theta}(y)\frac{\gamma^j}{j!},
\end{equation}
\smallskip where $f_\error(y)=f_{x,\theta}(y;0)$ is the probability density of $\error_j$ (since when $\gamma=0$,  $\widetilde Y_j=\error_j$) and
\begin{equation}\label{eq:alpha}
\alpha^j_{x,\theta}(y):=\frac1{f_\error(y)}\frac{\partial^j f_{x,\theta}}{\partial \gamma^j}(y;0),
\end{equation}
thus $\alpha^0_{x,\theta}(y)=1$. We note $f_{x,\theta}(y;\gamma)$ is in infinitely differentiable for all $y\in \R^L$, thus $\alpha^j_{x,\theta}(y)$ is always well-defined. We now use ~\eqref{eq:probdenstaylor} to give an expression of the $\chi^2$ divergence in terms of the moment tensors.

\begin{lem}\label{lem:falldominoes} 
	The divergence $\chi^2(f_{\tilde x,\tilde \theta}||f_{x,\theta})$ is expressed in terms of the moment tensors as:
	\begin{align}
	\nonumber \chi^2(f_{\tilde x,\tilde \theta}||f_{x,\theta})\hspace{-30pt}&\\
	&=\frac{\sigma^{-2d}}{(d!)^2}\E\left[\left(\alpha_{\tilde x,\tilde \theta}^d(\error)-\alpha_{x,\theta}^d(\error)\right)^2\right]+\OO(\sigma^{-2d-1}), \label{eq:falldominoes}\\
	&=\frac{\sigma^{-2d}}{d!}\|M^d_{\tilde x,\tilde \theta}-M^d_{x,\theta}\|^2+\OO(\sigma^{-2d-1}), \label{eq:chi2autocorr}
	\end{align}
	where $d=\inf\left\{n:\|M^n_{\tilde x,\tilde \theta}-M^n_{x,\theta}\|^2>0\right\}$.
\end{lem}

\begin{proof}
This proof is presented in more detail in \cite[Appendix B]{abbe2017multireference}. Equation~\eqref{eq:falldominoes} is obtained by Taylor expanding the $\chi^2$ divergence around $\gamma=0$, using \eqref{eq:probdenstaylor} and the fact that $\alpha_{\tilde x,\tilde \theta}^n(z)=\alpha_{x,\theta}^n(z)$ almost surely for all $n<d$, which follows from the definition of $d$ and equation \eqref{eq:chi2autocorr}. Now to prove \eqref{eq:chi2autocorr}, it is enough to show that 
\begin{equation}\label{eq:momentsvsalpha}
\E\left[\alpha_{\tilde x,\tilde \theta}^d(\error)\alpha_{x,\theta}^d(\error)\right]=d!\left<M^d_{\tilde x,\tilde \theta},M^d_{x,\theta}\right>,
\end{equation}    
Let $G$ and $\tilde G$ be two independent random variables such that $G\sim \theta$ and $\tilde G\sim \tilde \theta$. On one hand we have
\begin{align}
\left<M^d_{\tilde x,\tilde \theta},M^d_{x,\theta}\right>
&=\E\left[\left<{P\tilde G} \tilde x,PG x\right>^d\right]. \label{eq:GtildeG}
\end{align}
On the other hand, we can write $f_{x,\theta}$ explicitly by
\begin{align}
f_{x,\theta}(y)&=\E_G[f_\error(y- \gamma PG x)], \label{eq:gaussmixture}
\end{align}
where $G\sim \theta$, and using equation (\ref{eq:alpha}) we can write
\begin{align*}
&\hspace{-10pt}\E\left[\alpha_{\tilde x,\tilde \theta}^d(\error)\alpha_{x,\theta}^d(\error)\right]\\
&=\frac{\partial^{2d}}{\partial \tilde \gamma^d\partial \gamma^d}
\E\left[\frac{f_\error(\error- \tilde \gamma {P\tilde G} \tilde x)}{f_\error(\error)}\frac{f_\error(\error- \gamma PG x)}{f_\error(\error)}\right]_{\tilde \gamma,\gamma=0}\\
&=\E\left[\frac{\partial^{2d}}{\partial \tilde \gamma^d\partial \gamma^d}\exp\left(\gamma\tilde \gamma \left<{P\tilde G} \tilde x,PG x\right> \right)\right]_{\tilde \gamma,\gamma=0}\\
&=d!\,\E\left[\left<{P\tilde G} \tilde x,PG x\right>^d\right],
\end{align*}
where $G$ and $\tilde G$ are defined as in \eqref{eq:GtildeG}, and \eqref{eq:momentsvsalpha} finally follows from equation \eqref{eq:GtildeG}.
\end{proof}

\subsection{Final details of the proof of Theorem \ref{thm:corsigma}}\label{sec:FinalDetails}

By Theorem \ref{thm:ChapmanRobbins}, Lemma \ref{lem:falldominoes}, equations \eqref{eq:Cov} and \eqref{eq:nchi2} we obtain
\begin{equation}\label{eq:limitexplanation}
\MSE\ge \frac{\|\phi_{x}(\tilde x)-x\|^2}{\left(1+\sigma^{-2d}K_d+\OO\left(\sigma^{-2d-1}\right)\right)^N-1}.
\end{equation}
Equation \eqref{eq:momentssanov} now follows from
\begin{multline*}
\left(1+\sigma^{-2d}K_d+\OO(\sigma^{-2d-1})\right)^N=\\\exp\left(\lambda^{d}_N K^{d}_{\tilde x,\tilde \theta}\right)+\OO\left(\lambda^{d}_N \sigma^{-1}\right)
\end{multline*}
and taking the supremum over $\tilde x$ and $\tilde \theta$.

Finally we prove that the MLE is consistent, i.e. it converges to the true signal in probability, when $\rho=\infty$. Let

\begin{equation}\label{eq:DN}
\LN(\tilde x,\tilde \theta):=\frac{\sigma^{2\bar d}}{N}\sum_{i=1}^N\log \frac{f_{\tilde x,\tilde \theta}(\tilde Y_i)}{f_{x,\theta}(\tilde Y_i)}.
\end{equation}

The MLE is given by
$$\hat X_{\text{MLE}}=\argmax_{\tilde x}\max_{\tilde \theta}\LN(\tilde x,\tilde \theta).$$
Fix $\tilde x$ and $\tilde \theta$, and for ease of notation let $d=d_{\tilde x, \tilde \theta}$. We can write
$$\LN(\tilde x,\tilde \theta)=\frac{\sigma^{2(\bar d-d)}}{N/\sigma^{2d}}\sum_{i=1}^{N/\sigma^{2d}}\sum_{j=1}^{\sigma^{2d}}\log \frac{f_{\tilde x,\tilde \theta}(\tilde Y_{\sigma^{2d}(i-1)+j})}{f_{x,\theta}(\tilde Y_{\sigma^{2d}(i-1)+j})}$$
We have
\begin{align*}
\E\left[\sum_{j=1}^{\sigma^{2d}}\log \frac{f_{\tilde x,\tilde \theta}(\tilde Y_j)}{f_{x,\theta}(\tilde Y_j)}\right]&=\sigma^{2d}\E\left[\log \frac{f_{\tilde x,\tilde \theta}(\tilde Y)}{f_{x,\theta}(\tilde Y)}\right]\\
&=-\sigma^{2d}D(f_{\tilde x,\tilde \theta}||f_{x,\theta})
\end{align*}
where $D$ denotes the KL divergence, defined for two probability densities $f_A$ and $f_B$ as
\begin{equation*}
D(f_A||f_B):=
\E\left[\log \left(\frac{f_A(A)}{f_B(A)}\right)\right],
\end{equation*}
where $A\sim f_A$.
	
Using \eqref{eq:probdenstaylor}, with $\gamma=1/\sigma$, we have $f_{\tilde x,\tilde \theta}\rightarrow f_{x, \theta}$ as $\gamma\rightarrow 0$, which implies by \cite[Section F, Theorem 9]{sason2016f} that
$$\lim_{\gamma\rightarrow 0}\frac{D(f_{\tilde x,\tilde \theta}||f_{x,\theta})}{\chi^2(f_{\tilde x,\tilde \theta}||f_{x,\theta})}=\frac12,$$
and
$$D(f_{\tilde x,\tilde \theta}||f_{x,\theta})=\frac{\sigma^{-2d}}{2\, d!}\|M^d_{\tilde x,\tilde \theta}-M^d_{x,\theta}\|^2+\OO(\sigma^{-2d-1}),$$
thus by the law of large numbers, since $N/\sigma^{2d}$ diverges,
 $$\LN(\tilde x,\tilde \theta)\rightarrow\left\{\begin{array}{ll}
 -\infty&\mbox{if }d_{\tilde x, \tilde\rho}<\bar d\\
 -\frac{1}{2\, \bar d!}\|M^{\bar d}_{\tilde x,\tilde \theta}-M^{\bar d}_{x,\theta}\|^2&\mbox{otherwise }
 \end{array}\right.$$
 
As $N\rightarrow \infty$, the maximum of $\LN(\tilde x,\tilde \theta)$ tends to $0$ in probability, and is achieved when $\phi_x(\tilde x)=x$, thus the MLE must converge in probability to $gx$ for some $g\in \Theta$.

\section*{Acknowledgments}
EA was partly supported by the NSF CAREER Award CCF--1552131, ARO grant W911NF--16--1--0051 and NSF Center for the Science of Information CCF--0939370.  JP and AS were partially supported by Award Number R01GM090200 from the NIGMS, the Simons Foundation Investigator Award and Simons Collaborations on Algorithms and Geometry, the Moore Foundation Data-Driven Discovery Investigator Award, and AFOSR FA9550-17-1-0291.

We would like to thank Afonso Bandeira, Tamir Bendory, Joseph Kileel, William Leeb and Nir Sharon for many insightful discussions.

\bibliographystyle{ieeetr}
\bibliography{refs_gg}

\end{document}